\documentclass[11pt]{article}

\usepackage[top=1in, bottom=1in, left=1in, right=1in]{geometry}

\usepackage{comment}
\usepackage{amssymb}
\usepackage{amsmath}
\usepackage{amsthm}
\usepackage{color}
\usepackage{hyperref}
\usepackage{graphicx}
\usepackage{comment}
\usepackage[export]{adjustbox}
\usepackage{subcaption}

\usepackage{enumerate}
\usepackage{enumitem} 

\newtheorem{proposition}{Proposition}
\newtheorem{theorem}{Theorem}
\newtheorem{corollary}{Corollary}

\def\I{{\mathcal I}}

\def\I{{\mathcal I}}

\def\NP{${\mathcal{NP}}$}

\def\conv{\mathop{\rm conv}}
\def\size{\mathop{\rm size}}

\def\R{{\mathbb R}}
\def\Z{{\mathbb Z}}

\def\ie{{i.e.,} }

\def\N{$\mathcal{N}$}

\def\M{$\mathcal{M}$}
\def\CM{$\mathcal{C-M}$}
\def\PM{$\mathcal{PM}$}

\def\EC{$\mathcal{EC}$}
\def\CEC{$\mathcal{C-EC}$}

\def\SS{$\mathcal{SS}$}
\def\CSS{$\mathcal{C-SS}$}

\def\VC{$\mathcal{VC}$}
\def\CVC{$\mathcal{C-VC}$}
\def\PVC{$\mathcal{PVC}$}

\def\01{\ensuremath{0\mathord{-}1}}
\def\01{0-1}

\newcounter{mynotes}
\newcommand{\rednote}[1]{\addtocounter{mynotes}{1}{\textcolor{red}{$^{\arabic{mynotes}}$}}%
\marginpar{\scriptsize \textcolor{red}{ {\arabic{mynotes}.\ {\sf {#1}}}}}}
\newcommand{\cnote}[1]{\rednote{#1}}

\usepackage{lipsum}
\usepackage{titlesec}

\makeatletter
\def\thm@space@setup{%
  \thm@preskip=1ex plus 1ex minus .2ex
  \thm@postskip=\thm@preskip 
}
\makeatother

\titlespacing*{\section} {0pt}{2.5ex plus 1ex minus .2ex}{1.3ex plus .2ex}
\titlespacing*{\subsection} {0pt}{2.25ex plus 1ex minus .2ex}{0.5ex plus .2ex}
\titlespacing*{\subsubsection}{0pt}{2.25ex plus 1ex minus .2ex}{0.5ex plus .2ex}
\titlespacing*{\paragraph} {0pt}{1.25ex plus 1ex minus .2ex}{1em}
\titlespacing*{\subparagraph} {\parindent}{1.25ex plus 1ex minus .2ex}{1em}

\title{Totally Unimodular Congestion Games}

\author{Alberto Del Pia\thanks{Department of Industrial and Systems Engineering \& Wisconsin Institute for Discovery, University of Wisconsin-Madison. Email: \url{delpia@wisc.edu}}%
\and Michael Ferris\thanks{Department of Computer Sciences, University of Wisconsin-Madison. Email: \url{ferris@cs.wisc.edu}}%
\and Carla Michini\thanks{Wisconsin Institute for Discovery, University of Wisconsin-Madison. Email: \url{michini@wisc.edu}}%
}
\date{\today}

\begin{document}

\abovedisplayskip=6pt plus 3pt minus 9pt
\abovedisplayshortskip=0pt plus 3pt
\belowdisplayskip=6pt plus 3pt minus 9pt
\belowdisplayshortskip=3pt plus 3pt minus 4pt

\pagenumbering{gobble}

\maketitle


\begin{abstract}
We investigate a new class of congestion games, called \emph{Totally Unimodular (TU) Congestion Games},
where the players' strategies are binary vectors inside polyhedra defined by totally unimodular constraint matrices.
Network congestion games belong to this class.

In the symmetric case, when all players have the same strategy set, 
we design an algorithm that finds an optimal aggregated strategy and then decomposes it into the single players' strategies.
This approach yields strongly polynomial-time algorithms to (i) find a pure Nash equilibrium, and (ii) compute a socially optimal state, if the delay functions are weakly convex.
We also show how this technique  can be extended to matroid congestion games.

We introduce some combinatorial TU congestion games, where the players' strategies are matchings, vertex covers, edge covers, and stable sets of a given bipartite graph. 
In the asymmetric case, we show that for these games (i) it is PLS-complete to find a pure Nash equilibrium even in case of linear delay functions, and (ii) it is NP-hard to compute a socially optimal state, even in case of weakly convex delay functions.


\end{abstract}

\clearpage


\pagenumbering{arabic}

\section{Introduction}
\label{sec: intro}
A central problem of Algorithmic Game Theory concerns the existence of Nash equilibria and the design of polynomial-time algorithms for their computation.
Nash equilibria are a very powerful and well-studied solution concept, and they have had tremendous impact in economics and social sciences \cite{HolRot04}.
A \emph{mixed} strategy is a probability distribution over the \emph{pure} strategies of a player.
In his celebrated work \cite{Nash50,Nash51}, John Nash proved the existence of mixed equilibria for any game with a finite set of strategies.
In some applications mixed Nash equilibria seems to have no natural interpretation, and thus, pure Nash equilibria are sometimes a more suitable solution concept than mixed Nash equilibria \cite{Vet02,Ansetal04}. Unfortunately, pure Nash equilibria are not guaranteed to exist, and even if they do, they are often hard to compute.  

%
On the positive side, there are classes of games that are known to posses pure Nash equilibria. A prominent example are \emph{potential games}, originally introduced by Monderer and Shapley \cite{MonSha96}. The key property of potential games is the existence of a \emph{potential function}, i.e.~a function defined on the joint strategy set of the players, and such that, if any player unilaterally deviates from her strategy, the change in her payoff is equal to the change in the potential function.
A pure Nash equilibrium of a potential game can be found with a local search algorithm that minimizes the potential function over the strategy set of the game, where each iteration of the algorithm corresponds to an improving step of a player. Even if this algorithm is finite, it could take an exponential number of iterations to reach a local optimum, that is, a pure Nash equilibrium. Potential games are related to the complexity class \emph{Polynomial Local Search} (PLS) introduced by Johnson, Papadimitriou and Yannakakis \cite{JohPapYan88,SchYan91}.
In fact, this class includes all problems with a local search algorithm where each improving step can be performed in polynomial time.
Many families of potential games have been shown to be PLS-complete, suggesting that a polynomial-time algorithm to find a pure Nash equilibrium is unlikely to exist for these games in general.

In this work, we focus on \emph{congestion games}, a class of potential games that has been widely investigated in the literature (see, e.g.,~\cite{Ros73,RouTar,SkoVoc,FabPapTal04,Ros73bis,Car11,Cara12}), and we provide new insight on their computational complexity using a polyhedral approach.
In a congestion game, a set of resources is given, and each player selects a feasible subset of the resources in order to minimize her cost function. 
The cost of a player's strategy is the sum of the delays of the resources selected by the player, and the delay of each resource
is a function of
the total number of players using it.

%
An example are \emph{network congestion games}, where the resources are the arcs of a given digraph $D=(V,E)$ and the strategies of each player $i$ are all $(r^i,s^i)$-paths in $D$, for $r^i,s^i \in V$. Fabrikant et al.~\cite{FabPapTal04} gave an algorithm to find a pure Nash equilibrium in symmetric network congestion games, i.e.,~when all players share the same origin-destination pair. The algorithm is based on a reduction to minimum cost flow and runs in strongly polynomial time.
In the asymmetric case, network congestion games are PLS-complete, even in case of linear delay functions \cite{FabPapTal04,AckRogVoc08}.

Another class of symmetric congestion games for which a pure Nash equilibrium can be computed in polynomial time are \emph{matroid congestion games}, i.e.~congestion games where the strategy space of each player consists of the bases of a matroid over the set of resources. Ackermann et al.~\cite{AckRogVoc08} proved that for this class of games the best-response dynamics are polynomially bounded in the number of players and resources, and that the matroid property is also necessary for guaranteeing polynomial time convergence of the best-response dynamics to a Nash equilibrium. The work of Ackermann et al.~initiates the investigation of structural properties of congestion games that guarantee the polynomial time computability of a pure Nash equilibrium.
Their focus is on the convergence time for best responses, and on the local properties of the players' strategy spaces.
In contrast with this approach, we propose to exploit the global structure of the game and to study it from a polyhedral point of view.

\paragraph{Polyhedral approach.}
We study congestion games where the players' strategies are given implicitly through a polyhedral representation. 
For a congestion game with $N$ players and $n$ resources, let
$X^i = \{\chi^i \in \{0,1\}^n : \chi^i \text{ is the incidence vector of a strategy of player }i\}$,
and let $P^i = \{x^i \in [0,1]^n : A^i x^i \ge b^i\}$ be a polyhedron such that $P^i = \conv(X^i)$.
The game is symmetric if $A^i=A$ and $b^i=b$ for $i=1,\dots,N$.
%
%
Our main goal is to understand what properties of $P^i$ affect the computational complexity of finding a pure Nash equilibrium.
We investigate \emph{totally unimodular (TU) congestion games}, 
where $A^i$ is a $m \times n$ \emph{totally unimodular} (TU) matrix, i.e.~each square submatrix of $A^i$ has determinant equal to $0$, $+1$, or $-1$, and $b^i$ is integral, for all $i$.
It is easy to verify that network congestion games are of this form.
However, TU congestion games capture a larger array of congestion games. For example, consider a congestion game where the agents (e.g.~taxi drivers, call center operators) compete to supply their service to as many clients as possible in a finite time horizon.
Each time period is a resource that gets congested, since the profit of each agent in that period (the number of clients that each agent can serve) decreases with the total number of agents that are active.
Moreover, agents are typically not allowed to work in more than a certain number of consecutive time periods.
For example, taxi drivers may be constrained to work during at most a 8 hours time shift out of the $3$ shifts available in a day, and at least $5$ shifts per week. In this case, $n=3 \times 7$ is the number of shifts in a week, and $A$ is a binary matrix with the \emph{consecutive-ones property}, thus it is TU. It follows that this game can be modeled as a symmetric TU congestion game.


\paragraph{Our results.}
Our main contribution is to give an algorithm that runs in time polynomial in $n,m$ and $N$, to find a pure Nash equilibrium of symmetric TU congestion games defined using a TU matrix $A \in \R^{m\times n}$. 
Our algorithm is structured into two phases. In the first phase, we solve an aggregated problem, where we minimize the potential function over the \emph{aggregated strategies} of the players, to determine how many players should use each resource. Here, monotonicity of the delay functions and total unimodularity of the constraint matrix are crucial to reformulate the aggregated problem as a linear program.
In the second phase, we apply an algorithm that decomposes the optimal aggregated strategy into the single players' strategies. To this purpose, we rely on the so-called \emph{integer decomposition property} \cite{BauTro78}, holding for polyhedra defined by TU constraint matrices.

The same approach can be used to compute in strongly polynomial time a socially optimal state of a symmetric TU congestion game, if the delay functions are weakly convex.
Moreover, our two-phase algorithm can be applied also to (symmetric and asymmetric) matroid congestion games, and one of the crucial reasons why this is possible is that the independent set polytope of a matroid also has the integer decomposition property.
We further extend our algorithm to compute in strongly polynomial time a socially optimal state of a polymatroid congestion game, if the delay functions are weakly convex.

%

Concerning asymmetric congestion games, even if all the $P^i$s share the same TU constraint matrix $A$ and differ only in the right-hand side vectors $b^i$, $i=1,\dots,N$, we can directly show PLS-completeness of finding a pure Nash equilibrium and NP-hardness of finding a socially optimal state, since asymmetric network congestion games can be cast in this setting.
For a bipartite graph $G=(V,E)$ we introduce a number of combinatorial congestion games, where the strategies of each player are defined on a subgraph $G^i$ of $G$ as: $(i)$ matchings of $G^i$, $(ii)$ edge covers of $G^i$, $(iii)$ vertex covers of $G^i$, $(iv)$ stable sets of $G^i$. In the symmetric case, i.e., when $G^i = G$ for $i=1,\dots,N$, our algorithm runs in time polynomial in $|V|$, $|E|$, and $N$, and has combinatorial interpretations.
In the asymmetric case, we show  through a reduction from POS NAE 2SAT that all the above games are PLS-complete even in case of linear profit/delay functions. 
Moreover, if the delay functions are weakly convex, for all these combinatorial games it is NP-hard to compute a socially optimal state.

\section{Preliminaries}
\label{sec: definitions}
Let $(X,f)$ be a game with $N$ players, where $X^i$ is the \emph{strategy set} of player $i$, $X = X^1 \times \dots \times X^N$, $f^i: X \rightarrow \R$ is the \emph{cost function} of player $i$ and $f = (f^1, \dots, f^N)$. We assume that, for all $i \in \{1,\dots,N\}$, $X^i$ is a finite set, and we call each element $x \in X$ a \emph{state} of the game.
A \emph{pure Nash equilibrium} is a state $x = (x^1, \dots , x^N)$ such that for each $i$, $f^i(x^1,\dots,x^i,\dots,x^N) \le f^i(x^1,\dots,\bar{x}^i,\dots,x^N)$ for any $\bar{x}^i \in X^i$.

\paragraph{Congestion games.}
In a \emph{congestion game} there is a finite set of resources $R$ and, for each $i = 1,\dots, N$, a set $X^i \subseteq 2^R$. From now on we assume $|R| = n$ and we identify each $x^i \in X^i$ with its incidence vector in $\{0,1\}^n$.
A nondecreasing delay function $d_j : \{1,\dots,N\} \rightarrow \Z$
is associated with each resource $j \in \{1, \dots, n\}$. For a resource $j \in \{1, \dots, n\}$ and a state $x \in X$, denote by $t_j(x)$ the number of players using $j$ in state $x$.
The cost $f^i(x)$ incurred by player $i$ in state $x \in X$, is computed as the sum of the delays over all the resources selected by $i$, where the delay of resource $j$ is $d_j(t_j(x))$. Each player $i$, given the other players' strategies $x^{-i}$, chooses a strategy $x^i \in X^i$ that minimizes $f^i(x^i, x^{-i})$.
Note that we are not assuming the delays to be nonnegative. In particular, if the delays are negative, by turning the minimization into maximization problems and changing the signs, we obtain positive nonincreasing functions associated to the resources, to be interpreted as profits.
A congestion game is \emph{symmetric} if all $X^i$'s are the same, i.e. all players have the same strategy set.


In a classical paper \cite{Ros73}, Rosenthal proved that any congestion game has a pure Nash equilibrium. The proof is based on the fact that congestion games are a subclass of (exact) \emph{potential games}, and they admit the following \emph{potential function}:
\begin{equation}
\label{e: potential function}
\phi(x) = \sum_{j=1}^n \sum_{i=1}^{t_j(x)} d_j(i).
\end{equation}
Since $\phi$ is an \emph{exact} potential function, for any $i = 1,\dots,N$, $x= (x^i,x^{-i})\in X$, $\bar{x}=(\bar{x}^i,x^{-i}) \in X$, we have that $\phi(x) - \phi(\bar{x}) = f^i(x) - f^i(\bar{x})$. Thus a pure Nash equilibrium is a state $x \in X$ such that, for each player $i$ and $\bar{x}^i \in X^i$, $\phi(x^i, x^{-i}) \le \phi(\bar{x}^i,x^{-i})$. As a consequence, a global optimum of $\min \{\phi(x) : x \in X\}$ is a pure Nash equilibrium of the congestion game.

\paragraph{Totally unimodular congestion games.}
In \emph{totally unimodular (TU) congestion games}, the sets $X^i$ correspond to the vectors $x^i$ that satisfy:
\begin{align}
\label{e: TU congestion game}
A^i x^i \ge b^i \\
x^i \in \{0,1\}^n,\nonumber
\end{align}
where $b^i \in \Z^m$, and $A^i \in \{0,\pm 1\}^{m \times n}$ is a given TU matrix, as defined above and treated in detail in \cite{Sch86}.
A TU congestion game is \emph{symmetric} if $A^i = A$ and $b^i = b$ for every player $i=1,\dots,N$.

The class of TU congestion games widely extends the class of network games studied by Fabrikant et al.~\cite{FabPapTal04}.
Examples of TU matrices that are not incidence matrices of directed graphs include incidence matrices of bipartite graphs, matrices with the consecutive-ones property, and network matrices.
A full characterization of TU matrices is given by Seymour~\cite{Sey80}.
Moreover, Seymour's  algorithm can be used to recognize in polynomial time if a congestion game, where the strategy sets are given via systems of inequalities, is a TU congestion game or not.
\section{Symmetric TU congestion games}
\label{sec: poly}

\paragraph{Nash equilibria.}
We first investigate the computational complexity of finding a pure Nash equilibrium of a symmetric TU congestion game.

\begin{theorem}
\label{TU}
There is a strongly polynomial-time algorithm for finding a pure Nash equilibrium in symmetric TU congestion games.
\end{theorem}

\begin{proof}

The algorithm computes a global optimum of
\begin{equation}
\label{e: potential problem}
\min \{ \phi(x) : x \in X\},
\end{equation}
where $\phi$ is the potential function of the congestion game and $X = X^1 \times \dots \times X^N$ is the strategy space. Since a global optimum of \eqref{e: potential problem} is also a local optimum of the local search problem, the resulting state is a pure Nash equilibrium.

In the first phase of our algorithm, we set up an aggregated problem as follows. Since the value of the potential function \eqref{e: potential function} only depends on how many players use a resource, we sum up the constraints corresponding to a given row of matrix $A$ for all players, and we define the aggregated variables $z \in [0,N]^n \cap \Z^n$ such that $z_j = \sum_{i=1}^N x_j^i$ for each resource $j$. 

\begin{equation}
\label{e: potential problem aggr}
\min \left\{ \sum_{j=1}^n \sum_{i=1}^{z_j} d_j(i) : A z\ge Nb, \ z \in [0,N]^n \cap \Z^n\right\}.
\end{equation}

In the remainder of the proof, we find an optimum $\bar z$ of \eqref{e: potential problem aggr}, and then we show how to decompose it into a state $\bar x\in X$ such that  $\bar z = \bar x^1 + \dots + \bar x^N$, with $\bar x^i \in X^i$.
Since, for each state $x\in X$, the corresponding $z = x^1 + \dots + x^N$ is feasible for \eqref{e: potential problem aggr} and has the same objective value, we have that $\bar x$ is an optimum of \eqref{e: potential problem}.

Next, to model the objective function of \eqref{e: potential problem aggr} as a linear function, we introduce variables $y^i \in \{0,1\}^n$ for $i=1,\dots,N$, where $y_j^i = 1$ if at least $i$ players use resource $j$, and $y_j^i = 0$ otherwise. Since $z = y^1 + \dots + y^N$ for each resource $j$, we can write our \emph{aggregated problem} as:
\begin{align}
\label{agg}
\text{minimize} \quad & \sum_{j=1}^n \sum_{i=1}^N d_j(i) y_j^i \\
\text{subject to} \quad & \sum_{i=1}^N A y^i \ge Nb \nonumber \\
& 0 \le y^i \le 1 && i = 1,\dots, N \nonumber \\
& y^i \in \{0,1\}^n && i = 1,\dots, N. \nonumber
\end{align}
First, for each $z$ feasible for \eqref{e: potential problem aggr}, we define $y^i_j = 1$ if $i \le z_j$ and $y^i_j = 0$ otherwise. Note that $y$ is feasible for \eqref{agg} and with the same objective value as $z$. Moreover, for each $y$ feasible for \eqref{agg}, the vector $z = y^1 + \dots + y^N$ is feasible for \eqref{e: potential problem aggr} and has objective value not larger than that of $y$. Therefore, the optimal solution $\bar y$ of \eqref{agg} yields an optimal solution $\bar z$ of \eqref{e: potential problem aggr}.

We now show how to solve problem \eqref{agg}.
First, the constraint matrix of the aggregated problem \eqref{agg} is also TU, since it is of the form
\[
(
\underbrace{
\begin{array}{cccc}
A & A & \cdots & A
\end{array}}_{N \text{ times}}
).
\]
As the right hand side of the system is integral, the linear relaxation of the feasible set of the aggregated problem has only integral vertices (see for example Theorem 19.3 in \cite{Sch86}).
Thus we can find an optimal solution $\bar y$ of the aggregated problem via linear programming.
Using Tardos's \cite{Tar86} algorithm, this can be done in time polynomial in $\size(A)$ and in $N$, thus in time polynomial in $n,m,N$ because all entries of $A$ are in $\{0,\pm 1\}$.
The vector $\bar z =  \bar y^1 + \dots \bar y^N$ is then an optimum of \eqref{e: potential problem aggr}.
In the remainder of the proof, we show how to derive from $\bar z$ a state $\bar x$ with the same objective value in \eqref{e: potential problem}. 

Since $A \bar z \ge Nb$, and $0 \le \bar z_j \le N$ for every $j =1,\dots,n$, $\bar z$ is an integral vector in $N \cdot P$, where $P = \{x :  0 \le x \le 1 ,\ A x \ge b\}$.
Since $A$ is TU, the polyhedron $P$ has the \emph{integer decomposition property}, thus there exist integer vectors $\bar{x}^1,\dots,\bar{x}^N$ in $P$ such that $\bar z = \bar{x}^1+\dots+\bar{x}^N$. Following Baum and Trotter \cite{BauTro78}, we show how to obtain such vectors in strongly polynomial time.

We show how to find an integer vector $\bar x^1$ in $P$ such that $\bar z - \bar x^1$ is an integer vector in $(N-1) \cdot P$.
In order to do so, we define
\begin{equation}
\label{e: P1}
P^1 = \{ s : 0 \le s \le 1, \  \bar z - (N-1) \le s \le \bar z, \ b \le As \le A \bar z- (N-1) b \}.
\end{equation}
The polyhedron $P^1$ is nonempty since it contains $\bar z/N$, and is integral since its constraint matrix
is TU.
Again using Tardos's \cite{Tar86} algorithm, we can find a vertex $\bar x^1$ of $P^1$ in time polynomial in $n,m$.

By applying the above argument recursively $N$ times, we obtain integral vectors $\bar x^1, \bar x^2, \dots,\bar x^N$ in $P$ with $\bar z = \bar x^1 + \dots + \bar x^N$.
Therefore, $\bar x$ is a pure Nash equilibrium.
The total running time of the algorithm is polynomial in $n,m,N$.
\end{proof}

\paragraph{Socially optimal states.}
A related problem is to consider the complexity of computing a socially optimal state.
Define the \emph{social delay} of a state $x$ as 
\begin{equation}
\label{e: social delay}
\gamma(x) := \sum_{j=1}^n t_j(x) \cdot d_j(t_j(x)).
\end{equation}
A \emph{socially optimal state} is a state $x$ minimizing the social delay $\gamma(x)$.
We remark that, since we are dealing with congestion games, we are assuming that the delay functions $d_j$ associated with each resource $j \in \{1, \dots, n\}$ are nondecreasing.
A delay function $d_j$ is called \emph{weakly convex} if 
$$
i \cdot d_j(i) - (i - 1) \cdot d_j(i - 1) \le (i +1) \cdot d_j(i +1) - i \cdot d_j(i)
$$
for all $1 < i < N$. 
Weakly convex delay functions have been already studied in the context of socially optimal states.
Werneck et al.~\cite{WerSet00} consider spanning tree congestion games with such delay functions and show how to find efficiently a socially optimal solution.
Ackermann et al.~\cite{AckRogVoc08} extend their algorithm to matroid congestion games with weakly convex delay functions.
(See Section~\ref{sec: matroid} for more details on matroid congestion games).
For symmetric TU congestion games, under the assumption of weakly convex delay functions, the structure of the problem allows us to use techniques similar to the ones introduced in Theorem \ref{TU}. 

\begin{theorem}
\label{th: sym soc opt}
There exists a strongly polynomial time algorithm for the problem of computing a socially optimal state of a symmetric TU congestion game with weakly convex delay functions.
\end{theorem}

\begin{proof}
Let $(X,d)$ be a game with $N$ players.
By artificially setting $d_j(0)$ to be any scalar, we can rewrite the social delay of a given state $x$ as
\begin{align*}
\gamma(x) = \sum_{j=1}^n t_j(x) \cdot d_j(t_j(x)) 
= \sum_{j=1}^n \sum_{i=1}^{t_j(x)} \big(i \cdot d_j(i) - (i-1) \cdot d_j(i-1)\big).
\end{align*}
For each resource $j \in \{1, \dots, n\}$,
we define functions $d'_j : \{1,\dots,N\} \rightarrow \Z$ by 
$$d'_j(i) := i \cdot d_j(i) - (i-1) \cdot d_j(i-1).$$
The functions $d'_j$ are nondecreasing since the $d_j$ are weakly convex, thus we can consider them as new delay functions.
Consider now a new TU congestion game obtained from the original one by using the new delay functions $d'_j$.
By the proof of Theorem~\ref{TU}, we can find in strongly polynomial time a global minimum of 
$\min \{ \phi'(x) : x \in X\}$,
where $\phi'(x)$ is the potential function of the latter game.
This global minimum is a socially optimal state of the original game since
\begin{equation*}
\phi'(x) = \sum_{j=1}^n \sum_{i=1}^{t_j(x)} d'_j(i) = \gamma(x). \qedhere
\end{equation*}
\end{proof}
We remark that the weak convexity assumption in Theorem~\ref{th: sym soc opt} is necessary, since 
for general nondecreasing delay functions
computing a socially optimal state of a symmetric TU congestion game is NP-hard \cite{MeySch12}.

\section{Matroid and polymatroid congestion games}
\label{sec: matroid}
In this section we point out that the aggregation/decomposition algorithm presented in Section \ref{sec: poly} can be adapted to \emph{independent set matroid congestion games} and \emph{base matroid congestion games},
where the strategies of each player are, respectively, the independent sets and the bases of a matroid. We further extend our approach to \emph{polymatroid congestion games}. We remark that the results presented in this section hold for both the symmetric and the asymmetric case.

Formally, let $M_i = (R,\I_i)$, $i=1,\dots,N$, be an indexed family of matroids, and let $r_i$ denote the rank function of $M_i$. The matroids are all defined on the same ground set $R$, that is the set of congestible resources, and $|R|=n$.
In independent set matroid congestion games, $X^i$ consists of the incidence vectors of the independent sets  of $M_i$, for all $i = 1,\dots,N$, thus $P^i$ is the independent set polytope of $M_i$, defined as
\begin{equation}
\label{e: matroid}
P^i = \{x^i \in \R^n_+ : x^i(U) \le r_i(U) \; \forall \; U \subseteq R\},
\end{equation}
see \cite{Sch03}.
In base matroid congestion games we restrict to bases, thus $P^i$ contains the additional constraint $x^i(R) = r_i(R)$. 

A polytope in the form \eqref{e: matroid}, where the right-hand-side is given by a submodular, nondecreasing, normalized function $g_i: 2^R \rightarrow \Z_+$, is an integer polymatroid. We assume that $g_i$ is given through a value oracle. \emph{Polymatroid congestion games} are a generalization of independent set matroid congestion games where the strategies of each player are the integer vectors in an integer polymatroid, see \cite{HarKliPei14}. In \emph{base polymatroid congestion games}, for each $i=1,\dots, N$, we restrict to vectors such that $x^i(R) = g_i(R)$. Note that $x^i_j$ is now a nonnegative integer expressing the usage of resource $j$ by player $i$.
Thus, the cost of player $i$ in state $x$ is $\sum_{j=1}^n x^i_j d_j(t_j(x))$, where $t_j(x)$ is the total usage of resource $j$ in state $x$.
We refer to \cite{Sch03} for definitions and notions about matroids and polymatroids.

We first define, for any congestion game, a variant of the game where each player keeps only her maximum cardinality strategies. The next proposition will be useful to reduce base matroid and polymatroid congestion games to independent set matroid and polymatroid congestion games, respectively. See the Appendix for a proof.

\begin{proposition}
\label{reduction}
Let $C$ be a congestion game with $N$ players, resource set $R$, and strategy set $X^i$ for each player $i$.
Let $\bar C$ be another congestion game on $R$ with $N$ players, where the strategy set $\bar X^i$ of player $i$ consists only of the maximum cardinality subsets of $R$ in $X^i$.
Given an instance $I$ of $\bar C$, one can construct in strongly polynomial time an instance $h(I)$ of $C$ such that 
(i) a pure Nash equilibrium for $h(I)$ is also a pure Nash equilibrium for $I$, and
(ii) a socially optimal state for $h(I)$ is also a socially optimal state for $I$.
Moreover, the reduction only modifies the edge delays by a constant.
\end{proposition}

It is known that in base matroid congestion games players reach a Nash equilibrium after a polynomial number of best responses, see Theorem 2.5 in \cite{AckRogVoc08}. This immediately implies the next theorem, in the case of base matroid congestion games. 

\begin{theorem}
\label{t: matroid Nash}
There is a strongly polynomial-time algorithm for finding a pure Nash equilibrium in independent set and base matroid congestion games.
\end{theorem}

In the Appendix we give an alternative proof of Theorem \ref{t: matroid Nash} where we apply a variant of our aggregation/decomposition algorithm to find a global optimum of \eqref{e: potential problem}.
This algorithm can be extended to find a socially optimal state of a (base) polymatroid congestion game, if the delay functions are weakly convex.

\begin{theorem}
\label{t: polymatroid social}
There is a strongly polynomial time algorithm for the problem of computing a socially optimal state of a polymatroid or base polymatroid congestion game with weakly convex delay functions.
\end{theorem}
\begin{proof}
We prove the result for polymatroid congestion games. The statement on base polymatroid congestion games then follows by adapting Proposition \ref{reduction} to multisets.
The algorithm computes the global minimum of 
$\min \{ \gamma(x) : x \in X\}$, where $\gamma(x)$ is defined as in \eqref{e: social delay}.
The aggregated problem is
\begin{equation}
\label{e: social problem aggr polymatroid}
\min \left\{ \sum_{j=1}^n \gamma_j(z_j) : z(U) \le \sum_{i=1}^N g_i(U) \;\; \forall \; U \subseteq R, z \in \Z^n_+\right\},
\end{equation}
where, for all $j=1, \dots, n$, $\gamma_j(z_j) = z_j d_j(z_j)$ and $z_j = \sum_{i=1}^N x^i_j$.
Clearly, the polyhedron $Q$ obtained by replacing $z \in \Z^n_+$ with $z \in \R^n_+$ in the constraints of \eqref{e: social problem aggr polymatroid} defines an integer polymatroid, since $\sum_{i=1}^N g_i$ is an integer-valued set function that is submodular, nondecreasing and normalized.
Moreover, by assumption, $\gamma_j(z_j)$ is weakly convex in $z_j$. 
Since one can minimize a separable weakly convex function over a polymatroid with the greedy algorithm \cite{FedGro86}, we can solve the relaxation of \eqref{e: social problem aggr polymatroid} in strongly polynomial time. Moreover, as the right-hand-sides of the constraints in \eqref{e: potential problem aggr matroid} are integral, the optimal solution $\bar z$ returned by the greedy algorithm is integral.

By Theorem 44.6 and Corollary 46.2c in \cite{Sch03}, if $\bar z \in Q \cap \Z^n$, then $\bar z =  \bar x^1 + \cdots +  \bar x^N$, where $\bar x^i$ is in $P^i \cap \Z^n$ for $i=1,\dots,N$.
Moreover, the proof of Corollary 46.2c in \cite{Sch03} provides a strongly polynomial-time algorithm to find
the $N$ incidence vectors $\bar x^1, \dots, \bar x^N$. In each step, one finds an integer vector $x^i \in P^i$ such that $z - \sum_{k=1}^i x^k$ is an integer vector in $P^{i+1} + \cdots + P^N$.
Clearly $\bar x$ is a socially optimal state.
\end{proof}

A corollary of Theorem \ref{t: polymatroid social} is the existence of a strongly polynomial time algorithm to find a socially optimal state of a matroid congestion game with weakly convex delay functions. This result is stated for base matroid congestion games in Theorem 2.8 of \cite{AckRogVoc08}, as a generalization of Theorem 1 in \cite{WerSet00}. Since there is a reduction from integer polymatroids to matroids \cite{Hel74}, see also Section 44.6b in \cite{Sch03}, one could first reduce a polymatroid congestion game to a matroid congestion game, and then apply Theorem 2.8 in \cite{AckRogVoc08}. However, since for each $i$, we have to replicate any element $r$ of the ground set $g_i(r)$ times, the algorithm obtained through the reduction would be exponential in the size of $g_i(r)$.

\section{Combinatorial TU congestion games} 
\label{comb}

Next, we define five combinatorial congestion games that we will consider in the remainder of the paper.
(See~\cite{Sch03} for more details on the corresponding combinatorial problems.)
\begin{itemize}[noitemsep,nolistsep,leftmargin=1.1em]
\item
\emph{Network congestion games} (\N).
We are given a digraph $D=(V, E)$, with the arcs playing the role of the resources.  
For each player $i$, we are given two nodes $r^i, s^i \in V$, and 
the strategy set $X^i$ of player $i$ is the set of all directed paths in $D$ from $r^i$ to $s^i$.
\item
\emph{Matching congestion games} (\M) (resp.~\emph{edge cover congestion games} (\EC)).
We are given a graph $G=(V, E)$, with the edges playing the role of the resources.
For each player $i$, we are given a subgraph $G^i=(V^i,E^i)$ of $G$, and
the strategy set $X^i$ of player $i$ is the set of all matchings (resp.~edge covers) in $G^i$.
\item
\emph{Stable set congestion games} (\SS) (resp.~\emph{vertex cover congestion games} (\VC)).
We are given a graph $G=(V, E)$, with the nodes playing the role of the resources.
For each player $i$, we are given a subgraph $G^i=(V^i,E^i)$ of $G$, and 
the strategy set $X^i$ of player $i$ is the set of all stable sets (resp.~vertex covers) in $G^i$.
\end{itemize}


\begin{proposition}
\label{p: tu comb}
Congestion games \N, and congestion games \M, \EC, \SS, \VC\ on bipartite graphs are TU congestion games.
\end{proposition}
The proof of Proposition~\ref{p: tu comb} is given in the Appendix.
Recall that a TU congestion game is symmetric if $A^i = A$ and $b^i = b$ for every player $i=1,\dots,N$.
As a consequence, games \N\ are symmetric if all players have the same origin $r^i = r$ and destination $s^i = s$.
Games \M, \EC, \SS, \VC\ are symmetric if all players act on the same subgraph, i.e. $G^i = G$ for all $i=1,\dots,N$.
Theorem \ref{TU} then directly implies the following:
\begin{corollary}
\label{combgames}
There is a strongly polynomial-time algorithm for finding a pure Nash equilibrium in the symmetric case of games \N, and of games \M, \EC, \SS, \VC\ on bipartite graphs.
\end{corollary}

The algorithm given in the proof of Theorem \ref{TU} has some nice combinatorial interpretations for the above combinatorial games.
As an example, consider games \M.
Given bipartite graph $G = (V, E)$ and delays $d_e$, we construct a new bipartite graph $G' = (V,E')$ by replacing each edge $e \in E$ with $N$ parallel edges $e^1,\dots,e^N$ between the same nodes, with weights $-d_e(1), \dots, -d_e(N)$.
Solving problem \eqref{agg} is equivalent to finding a maximum weight subset $F$ of $E'$ such that each node $v$ in $V$ is incident to at most $N$ edges in $F$.
This is a \emph{simple $b$-matching} problem (see Ch.~21 in \cite{Sch03}).
Now let $\tilde G$ be the subgraph of $G$ obtained by deleting edge $e \in E$ if no edge among $e^1,\dots,e^N$ is in $F$.
Moreover, let $U$ be the set of nodes that have degree $N$ in $(V,F)$.
Finding an integer vector in \eqref{e: P1} is then equivalent to finding a matching in the graph $\tilde G$ covering all nodes in $U$.
Games \EC, \VC, and \SS\ have similar combinatorial interpretations.
For the symmetric case of \N, we recover the algorithm described in \cite{FabPapTal04}.

\smallskip
We now consider some variants of the combinatorial games previously defined that do not appear to be TU congestion games.
A \emph{maximum cardinality matching congestion game (\CM)} is a variant of \M\ where each set $X^i$ consists only of the maximum cardinality matchings in $G^i$.
Similarly, in a \emph{minimum cardinality edge cover congestion game (\CEC)} each set $X^i$ consists of the minimum cardinality edge covers in $G^i$; in a \emph{maximum cardinality stable set congestion game (\CSS)} each set $X^i$ consists of the maximum cardinality stable sets in $G^i$; in a \emph{minimum cardinality vertex cover congestion game (\CVC)} each set $X^i$ consists of the minimum cardinality vertex covers in $G^i$.

By Proposition \ref{reduction} and Corollary \ref{combgames} we obtain the following:

\begin{corollary}
There is a strongly polynomial-time algorithm for finding a pure Nash equilibrium in the symmetric case of the following games on bipartite graphs:
\CM,
\CEC,
\CSS,
\CVC.
\end{corollary}

Due to the wide applicability of the original combinatorial problems, also the combinatorial congestion games defined in this section have a large number of applications.
For example, consider a bipartite graph representing a road network with entry nodes and exit nodes.
Each player is an advertiser, who places her ad on minimum cardinality vertex cover of the graph, which is the cheapest way to expose each traveler to the ad.
In each node, the probability that a traveler sees a given ad decreases with respect to the total number of ads that have been placed there by all players. This yields nondecreasing delay functions associated to nodes.
Therefore each advertiser chooses a minimum cardinality vertex cover maximizing the expected number of times that her ad is seen by the travelers.

\subsection{The asymmetric case}
\label{sec: asymmetric}
In this section we focus on the asymmetric version of the combinatorial congestion games that we have just introduced. 
For all such combinatorial games on bipartite graphs, we show that 
it is PLS-complete to find pure Nash equilibria, 
even in the case of linear delay functions,
and that it is NP-hard to find socially optimal states,
even in the case of weakly convex delay functions.

We recall that potential games can be regarded as local search problems where the function to minimize is the potential function \eqref{e: potential function}, and the neighborhood $N(x)$ of a state $x \in X$ is the set of states arising from single player defections, i.e.,~$N(x) =\{(\bar x^i, x^{-i}) : \bar x^i \in X^i \setminus \{x^i\}, \ i\in\{1,\dots,N\}\}$.

\paragraph{PLS-completeness.}
Informally, a \emph{polynomial-time local search (PLS)} problem \cite{JohPapYan88} is a local search problem equipped with a polynomial-time algorithm that, given a solution $x$, either computes a better solution in $N(x)$, or determines that no such solution exists, meaning that $x$ is a \emph{local optimum}. This yields an algorithm that at each step moves from a solution to an improving neighboring solution. However this algorithm may still require an exponential number of steps to converge to a local optimum. A problem $\Pi$ is \emph{PLS-reducible} to a problem $\Pi'$ if any instance $\pi$ of $\Pi$ can be mapped in polynomial time to an instance $\pi'$ of $\Pi$, and any local optimum of $\pi'$ can be mapped in polynomial time to a local optimum of $\pi$. A problem in PLS is \emph{PLS-complete} if every problem in PLS is PLS-reducible to it. For formal definitions on PLS, we refer to \cite{JohPapYan88,SchYan91}.

A well-known PLS-complete problem is POS NAE 2SAT \cite{SchYan91}, i.e.~not-all-equal-2SAT with positive literals only: an instance consists of clauses in conjunctive normal form; each clause has at most two positive literals and is assigned a positive value; each clause is satisfied if its constituents do not all have the same value. A solution is a \emph{truth assignment}, i.e.~a $0/1$ assignment to all variables; the value of a solution, to be maximized, is the sum of the values of the satisfied clauses; the neighborhood of a solution contains all solutions obtained by flipping the value of one variable. The local search problem is to find a truth assignment whose value cannot be increased by flipping a variable.

We remark that asymmetric TU congestion games in the form \eqref{e: TU congestion game} are PLS-complete, even if $A = A^i$, $i=1,\dots,N$, and even in case of linear delay functions.
This follows directly from the PLS-completeness of asymmetric \N\ with linear delay functions~\cite{AckRogVoc08} (see also \cite{FabPapTal04}) and from the proof of Proposition~\ref{p: tu comb}.
%
%
Note that all the asymmetric variants of problems \M, \EC, \SS $\,$ and \VC $\,$ on bipartite graphs can be written as TU congestion games with $A = A^i$, for $i=1,\dots,N$. 
Our goal is to prove that all these asymmetric TU congestion games are PLS-complete. 

To this purpose we define the \emph{perfect matching congestion game} (\PM) as a variant of \M\ where for every player $i$, the subgraph $G^i$ admits a perfect matching, and where the set $X^i$ consists only of the perfect matchings in $G^i$.
An instance of \PM\ is an instance of \CM\ and of \CEC\, where all subgraphs $G^i$ admit a perfect matching.
Similarly, the \emph{perfect vertex cover congestion game} (\PVC) is a variant of \VC\ where for every player $i$, the subgraph $G^i$ admits a \emph{perfect vertex cover}, \ie a vertex cover that is also a stable set, and where the set $X^i$ consists only of the perfect vertex covers in $G^i$.
An instance of \PVC\ is an instance of \CVC\ and of \CSS\, where all subgraphs $G^i$ admit a perfect vertex cover.

\begin{theorem}
\label{t: perfect matching PLS}
It is PLS-complete to find a pure Nash equilibrium in the asymmetric versions of 
\PM,
\M,
\EC,
on bipartite graphs and even in the case of linear delay functions.
\end{theorem}

\begin{proof}
$(i)$ 
We give a PLS-reduction of POS NAE 2SAT to an asymmetric \PM\ on a bipartite graph $G$.
First, we define a map $h$ from any instance of POS NAE 2SAT to an instance of an asymmetric \PM\ on a bipartite graph $G$.
Given an instance $I$ of POS NAE 2SAT, we construct a congestion game $h(I)$ as follows. Denote by $C=\{c_1,\dots,c_n\}$ the clauses of $I$ and by $\{x_1,\dots,x_N\}$ its variables. Let $w_j$ be the value of clause $c_j$ and, for a truth assignment $x$ of $I$ denote by $w(x)$ the value of $x$. Each variable of POS NAE 2SAT is a player of the \PM\ and each NAE clause is a set of resources, i.e.~a set of edges. Precisely, for each NAE clause $c_j$ we build the graph ``gadget'' in Fig.~\ref{fig: matching}$(i)$. The graph gadget of clause $c_j$ is a 4-cycle $u_j,v_j,z_j,\bar{v}_j,u_j$.

Let $m_j \in \{1,2\}$ denote the number of variables in $c_j$. 
The edge delays are defined as follows:
\begin{itemize}[noitemsep,nolistsep]
\item If $e=\bar{v}_j u_j$ or $e=\bar{v}_j z_j$, then $d_e(i) = 0$ for $i=1,\dots,N$;
\item If $e=v_ju_j$ and $c_j$ contains at least a constant equal to $1$, then $d_e(i) = 0$ for $i=1,\dots,N$;
otherwise, $d_e(i) = w_j$ for $i=1,\dots,N$ if $m_j=1$, and $d_e(i) = w_j (i-1)$ if $m_j=2$;
\item If $e=v_j z_j$ and $c_j$ contains at least a constant equal to $0$, then $d_e(i) = 0$ for $i=1,\dots,N$; 
otherwise, $d_e(i) = w_j$ for $i=1,\dots,N$ if $m_j=1$, and $d_e(i) = w_j (i-1)$ if $m_j=2$.
\end{itemize}
Now, we build graph $G$ as follows: for any two clauses $c_j$ and $c_{j+1}$, $j=1,\dots,n-1$ we identify $z_j$ and $u_{j+1}$; moreover, we identify $z_n$ and $u_{1}$, see Fig.~\ref{fig: matching}$(ii)$.
Clearly, $G$ is a bipartite graph with bipartitions $\{v_j,\bar{v}_j\}_{j=1,\dots,n}$ and $\{u_j\}_{j=1,\dots,n}$. For each $i=1,\dots,N$, let $C(i)$ denote the set of clauses containing variable $x_i$ and let $V_i = \{u_j,z_j:j=1,\dots,n\} \cup \{v_j : c_j \in C(i)\} \cup \{\bar{v}_j  : c_j \notin C(i)\}$. We assign to player $i$ the subgraph $G_i$ of $G$ induced by nodes in $V_i$.
This shows how to map $I$ to an instance $h(I)$ of asymmetric \PM\ on a bipartite graph.

Next, we define a map $g$ from states of $h(I)$ to truth assignments of $I$. A state of $h(I)$ is a set of $N$ perfect matchings on graphs $G_i$, $i=1,\dots,N$.
Note that each subgraph $G_i$ is a cycle of length $2n$  that admits two perfect matchings:
$M^i_0 = \{u_j v_j : c_j \in C(i)\} \cup \{u_j \bar{v}_j : c_j \notin C(i)\}$ and 
$M^i_1 = \{v_j z_j  : c_j \in C(i)\} \cup \{\bar{v}_j z_j: c_j \notin C(i)\}$.
Let $g_i : \{M^i_0,M^i_1\} \rightarrow \{0,1\}$ such that $g_i(M^i_0) = 0$ and $g_i(M^i_1) = 1$. We map strategy $M^i \in \{M^i_0,M^i_1\}$ of player $i$ to $x_i = g_i(M^i)$.
Setting $g = (g_1,\dots,g_N)$ shows that any state of $h(I)$ is mapped to a truth assignment of $I$, and that the mapping is bijective.

Finally, we need to show that any pure Nash equilibrium of $h(I)$ maps to a local optimum of $I$.
First, we remark that, for any truth assignment $x$ of $I$, each $x' \in N(x)$
obtained by flipping variable $x_i$, $i \in \{1,\dots,N\}$ is in one-to-one correspondence with the state obtained from $g^{-1}(x)$ after the defection of player $i$.
Now, let $M^1,\dots,M^N$ be a pure Nash equilibrium of $h(I)$, where $M^i \in \{M^i_0,M^i_1\}$ for all $i=1,\dots,N$, and denote by $y$ the corresponding state. Then, for any state $y'$ obtained from $y$ by switching perfect matching $M^i$, we have that $f_i(y') - f_i(y) \ge 0$, where $f_i$ denotes the cost of player $i$ in \PM.
By construction, for $x = g(y)$ and $x' = g(y')$ we have that
$w(x) - w(x') =  f_i(y') - f_i(y) \ge 0$.
This proves that $x$ is a local optimum of $I$.
%
%

$(ii)$ 
Since \PM\ is equivalent to \CM\ when all subgraphs $G^i$ admit a perfect matching, $(i)$ implies that asymmetric \CM\ is PLS-complete on a bipartite graph. 
The result then follows by Proposition \ref{reduction}.

$(iii)$ 
Since \PM\ is equivalent to \CEC\ when all subgraphs $G^i$ admit a perfect matching, $(i)$ implies that asymmetric \CEC\ is PLS-complete on a bipartite graph. 
The result then follows by Proposition \ref{reduction}.
\end{proof}
We can show analogous result for the asymmetric \PVC, \VC, and \SS\ on a bipartite graph, see the Appendix for a proof.
\begin{theorem}
\label{t: perfect vertex cover PLS}
It is PLS-complete to find a pure Nash equilibrium in the asymmetric versions of 
\PVC,
\VC,
\SS,
on bipartite graphs and even in the case of linear delay functions.
\end{theorem}
%



\paragraph{NP-hardness.}

Computing a socially optimal state for asymmetric TU congestion games is NP-hard, even if $A = A^i$, $i=1,\dots,N$.
This follows from the NP-hardness of asymmetric \N\ \cite{MeySch12}.
The next two theorems show that finding socially optimal states in the combinatorial games that we have introduced is NP-hard.
The proof of Theorem~\ref{th: NP-hard 1} is given in the Appendix, and is based on a reduction similar to the one given in the proof of Theorem~\ref{t: perfect matching PLS}.

\begin{theorem}
\label{th: NP-hard 1}
It is NP-hard to find a socially optimal state in the asymmetric versions of 
\PM,
\M,
\EC,
on bipartite graphs and even in the case of weakly convex delay functions.
\end{theorem}


\begin{theorem}
\label{t: perfect vertex cover NPH}
It is NP-hard to find a socially optimal state in the asymmetric versions of
\PVC,
\VC,
\SS,
on bipartite graphs and even in the case of weakly convex delay functions.
\end{theorem}

\newpage

\section*{Appendix}
\label{sec: appendix}

\begin{proof}[Proof of Proposition \ref{reduction}]
(i) Let $I$ be an instance of $\bar C$.
Let $\bar d_j(i)$, $i=1,\dots,N$, be the given delays in $I$, and let $\Delta := \max\{|\bar d_j(i)| : j \in R, \ i=1,\dots,N\}$.
We set the delays in $h(I)$ to be $d_j(i)=\bar d_j(i) - (2|R| \Delta + 1)$. 
It can be checked that for every player $i$, the cost in $h(I)$ corresponding to a strategy that is of maximum cardinality will always be strictly smaller than the cost in $h(I)$ corresponding to a strategy that is not of maximum cardinality.

We show that a pure Nash equilibrium $x$ for $h(I)$ is also a pure Nash equilibrium for $I$.
The strategy of each player $i$ in $x$ is of maximum cardinality, thus $x$ is also a state of $I$.
The cost $f^i(x)$ of player $i$ in $h(I)$ is $f^i(x) = \bar f^i(x) - k^i (2|R| \Delta + 1)$, where $\bar f^i(x)$ is the cost of player $i$ in $I$, and $k^i$ is the cardinality of a maximum cardinality strategy of player $i$.
By contradiction, assume that there is another strategy $\tilde x^i$ for player $i$ in $I$ such that in the state $\tilde x$ obtained from $x$ by swapping $x^i$ with $\tilde x^i$, player $i$ has lower cost in $I$, \ie $\bar f^i(\tilde x) < \bar f^i(x)$.
Clearly $\tilde x^i$ is a another maximum cardinality strategy.
This implies $f^i(\tilde x) = \bar f^i(\tilde x) - k^i (2|R| \Delta + 1) < \bar f^i(x) - k^i (2|R| \Delta + 1) = f^i(x)$, contradicting the fact that $x$ is a pure Nash equilibrium for $h(I)$.

(ii) Let $I$ be an instance of $\bar C$.
Let $\bar d_j(i)$, $i=1,\dots,N$, be the given delays in $I$, and let $\Delta := \max\{|\bar d_j(i)| : j \in R, \ i=1,\dots,N\}$.
We set the delays in $h(I)$ to be $d_j(i)=\bar C_j(i) - (2N|R| \Delta + 1)$. 
It can be checked that the social delay in $h(I)$ corresponding to a state where each strategy is of maximum cardinality will always be strictly smaller than the social delay in $h(I)$ corresponding to a state where at least one strategy is not of maximum cardinality. 

We show that a socially optimal state $x$ for $h(I)$ is also a socially optimal state for $I$.
The strategy of each player $i$ in $x$ is of maximum cardinality, thus $x$ is also a state of $I$.
The social delay $\gamma(x)$ of $x$ in $h(I)$ is $\gamma(x) = \bar\gamma(x) - (\sum_{j=1}^n t_j(x)) (2N|R| \Delta + 1)$, where $\bar\gamma(x)$ is the social delay of $x$ in $I$.
By contradiction, assume that there is another state $\tilde x$ with lower social delay, \ie $\bar\gamma(\tilde x) < \bar\gamma(x)$.
Clearly $\tilde x$ is a state where the strategy of each player is of maximum cardinality, thus $\sum_{j=1}^n t_j(\tilde x) = \sum_{j=1}^n t_j(x)$.
This implies 
$$\gamma(\tilde x) = \bar\gamma(\tilde x) - (\sum_{j=1}^n t_j(\tilde x)) (2N|R| \Delta + 1) < \bar\gamma(x) - (\sum_{j=1}^n t_j(x)) (2N|R| \Delta + 1) = \gamma(x),$$
contradicting the fact that $x$ is a socially optimal state for $h(I)$.

Note that both the given reductions are strongly polynomial.
\end{proof}

\medskip

\begin{proof}[Proof of Proposition \ref{p: tu comb}]
Network congestion games are TU congestion games, since $x^i$ is the incidence vector of a dipath from $r^i$ to $s^i$ in digraph $D=(V, E)$ if and only if $A x^i = b^i, \ x^i \in \{0,1\}^E$, where $A$ is the $V \times E$ incidence matrix of $D$, and $b^i$ has entry $-1$ corresponding to node $r^i$, entry  $+1$ corresponding to node $s^i$, and all other entries are $0$.

\M\ and \EC\ on bipartite graphs are TU congestion games, since $x^i$ is the incidence vector of a matching (resp.~edge cover) in $G^i = (V^i,E^i)$ if and only if $A^i x^i \le 1$ (resp.~$A^i x^i \ge 1$), $x^i \in \{0,1\}^{E^i}$, where $A^i$ is the $V^i \times E^i$ incidence matrix of $G^i$.

\SS\ and \VC\ on bipartite graphs are TU congestion games, since $x^i$ is the incidence vector of a stable set (resp.~vertex cover) in $G^i=(V^i,E^i)$ if and only if $A^i x^i \le 1$ (resp.~$A x^i \ge 1$), $x^i \in \{0,1\}^{V_i}$, where $A^i$ is the $E^i \times V^i$ incidence matrix of $G^i$.
\end{proof}

\begin{proof} [Proof of Theorem \ref{t: matroid Nash}]
We first consider independent set matroid congestion games.
Again, the algorithm computes the global optimum of \eqref{e: potential problem}.
The aggregated problem \eqref{e: potential problem aggr} becomes
\begin{equation}
\label{e: potential problem aggr matroid}
\min \left\{ \sum_{j=1}^n \phi_j(z_j) : z(U) \le \sum_{i=1}^N r_i(U) \;\; \forall \; U \subseteq R, z \in \Z^n_+\right\},
\end{equation}
where, for all $j=1, \dots, n$, $\phi_j(z_j) = \sum_{i=1}^{z_j} d_j(i)$. Note that $\phi_j(z_j)$ is weakly convex in $z_j$, since $d_j(i)$ is nondecreasing in $i$. Moreover, the polyhedron $Q$ obtained by replacing $z \in \Z^n_+$ with $z \in \R^n_+$ in the constraints of \eqref{e: potential problem aggr matroid} defines an integer polymatroid, since $\sum_{i=1}^N r_i$ is an integer-valued set function that is submodular, nondecreasing and normalized.

Since one can minimize a separable weakly convex function over a polymatroid with the greedy algorithm \cite{FedGro86}, we can solve the relaxation of \eqref{e: potential problem aggr matroid} in strongly polynomial time. Moreover, as the right-hand-sides of the constraints in \eqref{e: potential problem aggr matroid} are integral, the optimal solution $\bar z$ returned by the greedy algorithm is integral.

It is known that the independent set polytope of a matroid has the integer decomposition property (see \cite{Sch03}, Corollary 42.1e), and this property can also be generalized to the asymmetric case. In fact, if $\bar z \in Q \cap \Z^n$, then $\bar z =  \bar x^1 + \cdots +  \bar x^N$, where $\bar x^i$ is in $P^i \cap \Z^n$ for $i=1,\dots,N$, see Theorem 44.6 and Corollary 46.2c in \cite{Sch03}.
Moreover, the proof of Corollary 46.2c in \cite{Sch03} provides a polynomial-time algorithm to find
the $N$ incidence vectors $\bar x^1, \dots, \bar x^N$. Alternatively, one could apply Edmond's matroid union algorithm \cite{Edm68}.
Clearly $\bar x$ is a pure Nash equilibrium.

The statement on base matroid congestion games follows directly from Proposition \ref{reduction}.
\end{proof}

\medskip

\begin{proof}[Proof of Theorem \ref{t: perfect vertex cover PLS}]
$(i)$ 
We give a PLS-reduction of POS NAE 2SAT to an asymmetric \PVC\ on a bipartite graph $G$. The proof structure is similar to that of Theorem \ref{t: perfect matching PLS}, thus here we only outline the main differences. Again, the variables of POS NAE 2SAT map to players of \PVC\ and clauses map to a set of resources.
For each NAE clause $c_j$ we build the graph ``gadget'' in Fig.~\ref{fig: vc}$(i)$, that is a 8-cycle $u_j,s_j,v_j,t_j,z_j,\bar{t}_j,\bar{v}_j,\bar{s}_j,u_j$.

For $m_j \in \{1,2\}$ the vertex delays are defined as:
\begin{itemize}[noitemsep,nolistsep]
\item If $u \notin \{s_j,v_j\}$, then $d_u(i) = 0$ for $i=1,\dots,N$;
\item If $u = s_j$ and $c_j$ contains at least a constant equal to $1$, then $d_u(i) = 0$ for $i=1,\dots,N$; 
otherwise, $d_e(i) = w_j$ for $i=1,\dots,N$ if $m_j=1$, and $d_e(i) = w_j (i-1)$ if $m_j=2$;
\item If $u = v_j$ and $c_j$ contains at least a constant equal to $0$, then $d_u(i) = 0$ for $i=1,\dots,N$; 
otherwise, $d_e(i) = w_j$ for $i=1,\dots,N$ if $m_j=1$, and $d_e(i) = w_j (i-1)$ if $m_j=2$.
\end{itemize}
We build a bipartite graph $G$ by identifying $z_n$ and $u_{1}$ and $z_j$ and $u_{j+1}$ for $j=1,\dots,n-1$, see Fig.~\ref{fig: vc}$(ii)$.
For each $i=1,\dots,N$, we let $V_i = \{u_j,z_j: j=1,\dots,n\} \cup \{s_j,v_j,t_j : c_j \in C(i)\} \cup \{\bar{s}_j,\bar{v}_j,\bar{t}_j  : c_j \notin C(i)\}$ and we define $G_i$ as the subgraph of $G$ induced by nodes in~$V_i$.

Since each
$G_i$ is a cycle of length $4n$, it admits exactly two perfect vertex covers:
$W^i_0 = \{s_j, t_j  : c_j \in C(i)\} \cup \{\bar{s}_j, \bar{t}_j: c_j \notin C(i)\}$ and 
$W^i_1 = \{u_j, v_j : c_j \in C(i)\} \cup \{u_j, \bar{v}_j : c_j \notin C(i)\}$.
We define $g_i : \{W^i_0,W^i_1\} \rightarrow \{0,1\}$ such that $g_i(W^i_0) = 0$ and $g_i(W^i_1) = 1$ and we map strategy $W^i \in \{W^i_0,W^i_1\}$ of player $i$ to $x_i = g_i(W^i)$.

Let $y$ be a Nash equilibrium of \PVC\ corresponding to vertex covers $W^1,\dots,W^N$, and denote by $f_i$ the cost function of player $i$.
Then, for any state $y'$ obtained by switching perfect vertex cover $W^i$, we have that $f_i(y') - f_i(y) \ge 0$.
By construction, for $x = g(y)$ and $x' = g(y')$ it follows that $w(x) - w(x') =  f_i(y') - f_i(y) \ge 0$.

\smallskip

$(ii)$ 
Since \PVC\ is equivalent to \CVC\ when all subgraphs $G^i$ admit a perfect vertex cover, $(i)$ implies that asymmetric \CVC\ is PLS-complete on a bipartite graph. The result then follows by Proposition \ref{reduction}.
\smallskip

$(iii)$ 
Since \PVC\ is equivalent to \CSS\ when all subgraphs $G^i$ admit a perfect vertex cover, $(i)$ implies that asymmetric \CSS\ is PLS-complete on a bipartite graph. The result then follows by Proposition~\ref{reduction}.
\end{proof}

\medskip

\begin{proof}[Proof of Theorem~\ref{th: NP-hard 1}]
$(i)$ 
We give a polynomial reduction of POS NAE 3SAT to an asymmetric \PM\ on a bipartite graph $G$.
\emph{POS NAE 3SAT} is a \NP-complete variant of NAE 3SAT in the absence of negated variables (see \cite{Sch78}).

First, we define a map $h$ from any instance of POS NAE 3SAT to an instance of an asymmetric \PM\ on a bipartite graph $G$.
Given an instance $I$ of POS NAE 3SAT, we construct a congestion game $h(I)$ as follows. Denote by $C=\{c_1,\dots,c_n\}$ the clauses of $I$ and by $\{x_1,\dots,x_N\}$ its variables. 
Each variable of POS NAE 3SAT is a player of the \PM\ and each NAE clause is a set of resources, i.e.~a set of edges. Precisely, for each NAE clause $c_j$ we build the graph ``gadget'' in Fig.~\ref{fig: matching}$(i)$. The graph gadget of clause $c_j$ is a 4-cycle $u_j,v_j,z_j,\bar{v}_j,u_j$.

The weakly convex edge delays are defined as follows:
\begin{itemize}[noitemsep,nolistsep]
\item If $e=\bar{v}_j u_j$ or $e=\bar{v}_j z_j$, then $d_e(i) = 0$ for $i=1,\dots,N$;
\item If $e=v_ju_j$ or $e=v_j z_j$, then $d_e(1) = 0$, and $d_e(i) = i-2$ for $i = 2,\dots,N$.
\end{itemize}

Following the proof of Theorem~\ref{t: perfect matching PLS},
we map $I$ to an instance $h(I)$ of asymmetric \PM\ on a bipartite graph $G$, and 
we define a map $g$ from states of $h(I)$ to truth assignments of $I$.

Finally, we need to show that a state of $h(I)$ has social delay equal to zero if and only if it maps to a truth assignment $I$ that satisfies the POS NAE 3SAT formula.
Let $M^1,\dots,M^N$ be a state of $h(I)$, where $M^i \in \{M^i_0,M^i_1\}$ for all $i=1,\dots,N$.
Then $h(I)$ has social delay equal to zero if and only if for every clause $c_j$ there are at most two players $i \in \{1,\dots,N\}$ with $v_j u_j \in M^i$, and 
at most two players $i \in \{1,\dots,N\}$ with $v_j z_j \in M^i$.
This happens if and only if for every clause $c_j$ there are at most two variables $x_i$ contained in the clause and with value zero, and there are at most two variables $x_i$ contained in the clause and with value one.
Clearly, this happens if and only if $I$ satisfies the POS NAE 3SAT formula.


$(ii)$ 
Since \PM\ is equivalent to \CM\ when all subgraphs $G^i$ admit a perfect matching, $(i)$ implies that asymmetric \CM\ is NP-hard on a bipartite graph. 
Since modifying delay functions by a constant maintains weak convexity, the result then follows by Proposition \ref{reduction}.

$(iii)$ 
Since \PM\ is equivalent to \CEC\ when all subgraphs $G^i$ admit a perfect matching, $(i)$ implies that asymmetric \CEC\ is NP-hard on a bipartite graph. 
Since modifying delay functions by a constant maintains weak convexity, the result then follows by Proposition \ref{reduction}.
\end{proof}

\newpage

\section*{Figures}

\begin{figure}[h]
\begin{tabular}{cc}
  \def\svgwidth{.19\columnwidth}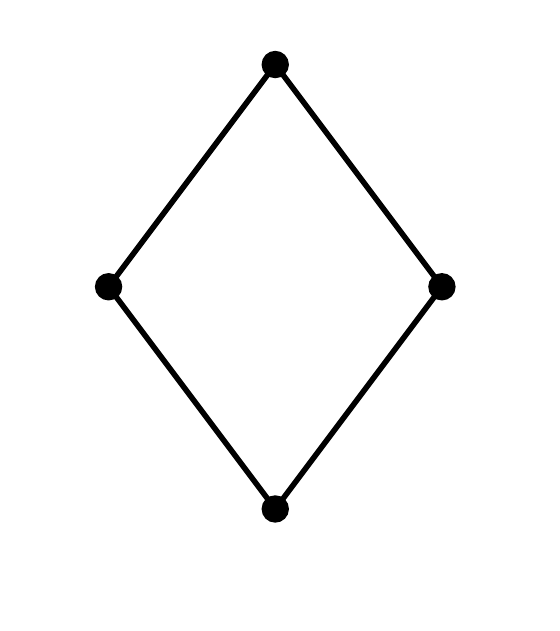 & \qquad \qquad
  \def\svgwidth{0.45\columnwidth}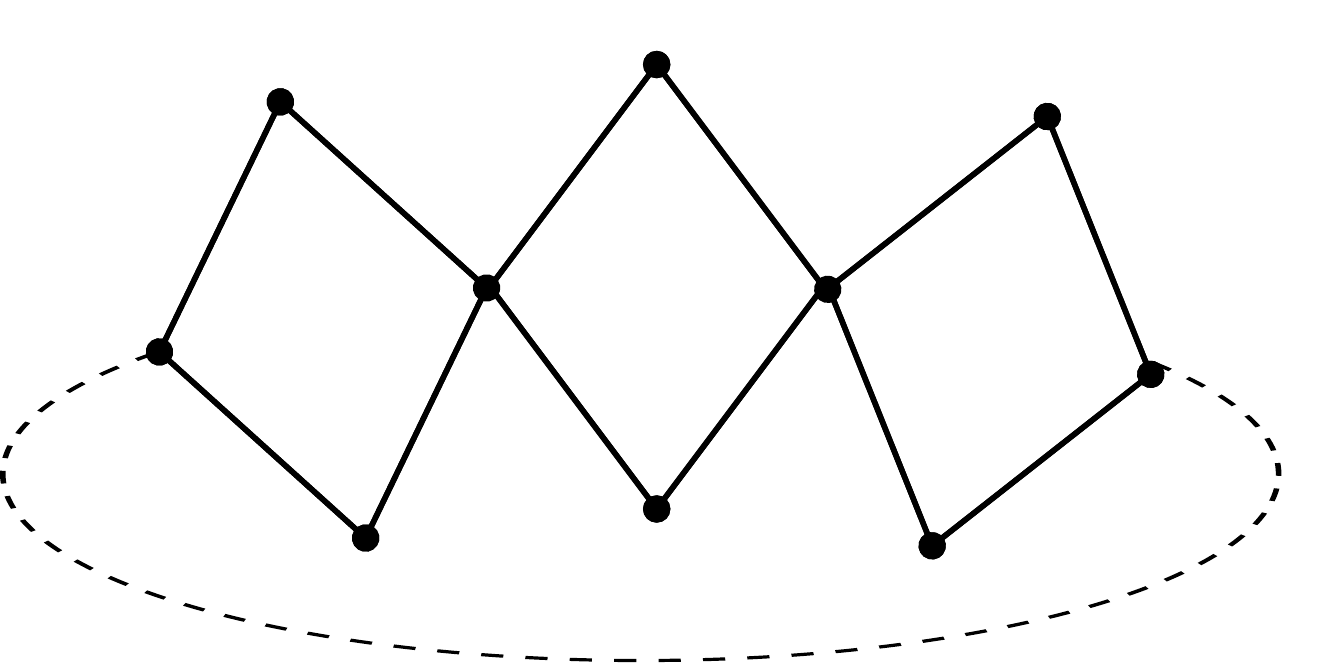\\
$(i)$ & $(ii)$
\end{tabular} 
\caption{Reduction from POS NAE 2SAT to asymmetric \PM\ on a bipartite graph.}
\label{fig: matching}
\end{figure}

\begin{figure}[h]
\begin{tabular}{cc}
  \def\svgwidth{.19\columnwidth}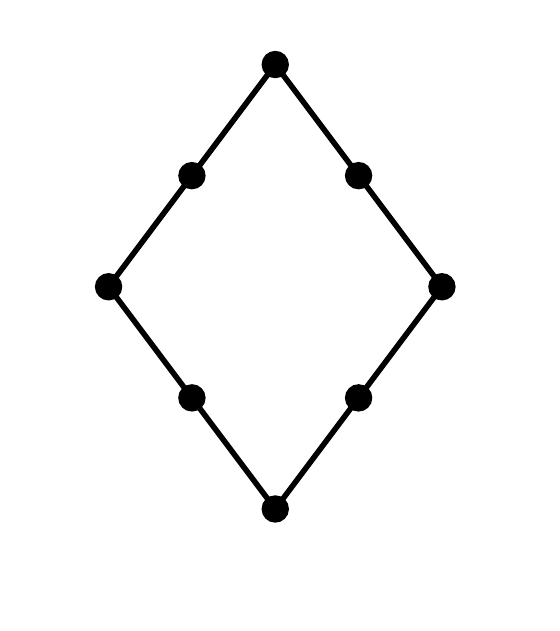 & \qquad \qquad
  \def\svgwidth{0.45\columnwidth}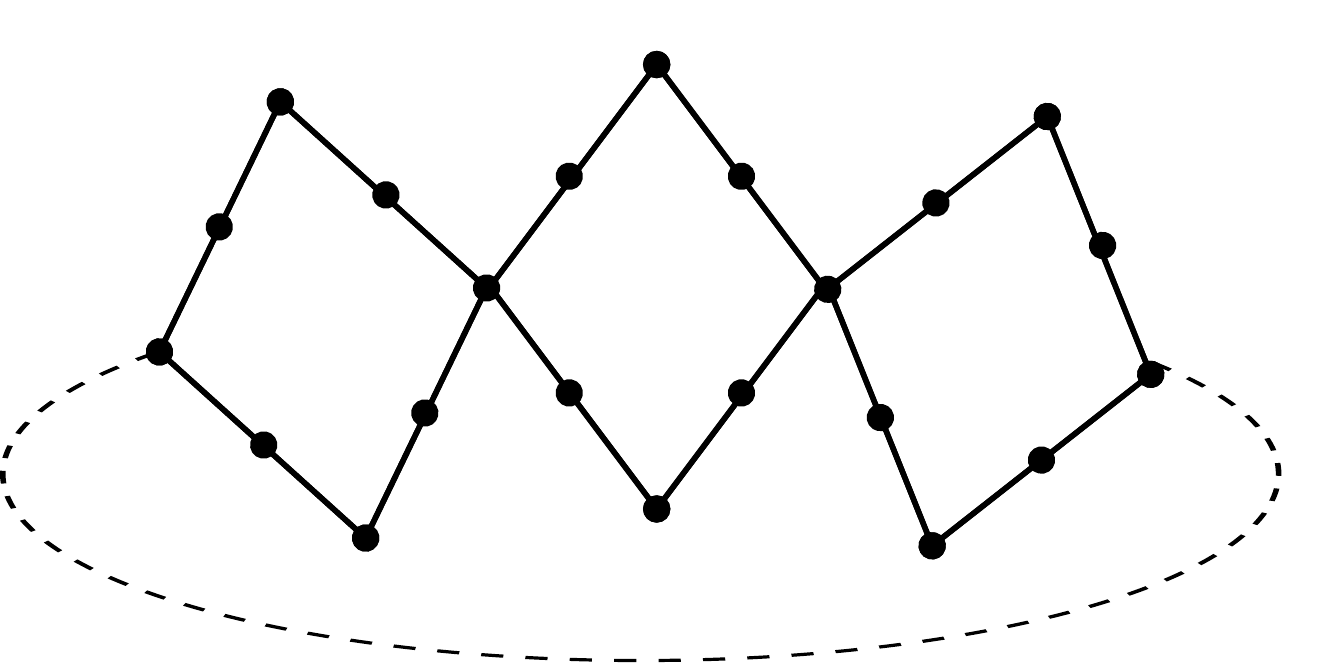\\
$(i)$ & $(ii)$
\end{tabular} 
\caption{Reduction from POS NAE 2SAT to asymmetric \PVC\ on a bipartite graph.}
\label{fig: vc}
\end{figure}

\newpage

\bibliographystyle{plain}
\bibliography{biblio}

\end{document}